\documentclass[12pt,a4paper]{article}

\usepackage{fancyhdr}
\usepackage{amsmath}
\usepackage{bm}
\usepackage{esint}
\usepackage{amsfonts}
\usepackage{amsthm}
\usepackage{amssymb}
\usepackage{amsbsy}
\usepackage{verbatim}
\usepackage{graphicx}
\usepackage{url}
\usepackage{mathrsfs}
\usepackage[hmargin = 1in,vmargin=.75in]{geometry}
\usepackage{color}
\usepackage{amstext}
\usepackage{stmaryrd}
\usepackage{enumerate}
\usepackage{algpseudocode}
\usepackage{algorithm}
\usepackage{pifont}
\usepackage{prettyref}
\usepackage{titling}
\usepackage{authblk}

\font\msbm=msbm10

\numberwithin{equation}{section}

\theoremstyle{plain}
\newtheorem{theorem}{Theorem}[section]
\newtheorem{lemma}[theorem]{Lemma}

\newtheorem{definition}{Definition}[section]

\newtheorem{remark}[theorem]{Remark}
\def\mathbb#1{\hbox{\msbm{#1}}}

\newcommand{\bq}{\boldsymbol{q}}

\newcommand{\bz}{\boldsymbol{z}}
\newcommand{\bone}{\boldsymbol{1}}

\newcommand{\BA}{\boldsymbol{A}}

\newcommand{\BJ}{\boldsymbol{J}}

\newcommand{\BQ}{\boldsymbol{Q}}

\newcommand{\BX}{\boldsymbol{X}}

\newcommand{\BZ}{\boldsymbol{Z}}

\newcommand{\btheta}{\boldsymbol{\theta}}

\newcommand{\bomega}{ \boldsymbol{\omega}}

\newcommand{\I}{\boldsymbol{I}}
\newcommand{\RR}{\mathbb{R}}
\newcommand{\lag}{\langle}
\newcommand{\rag}{\rangle}

\newcommand{\Tr}{\text{Tr}}

\newcommand*\diff{\mathop{}\!\mathrm{d}}

\DeclareMathOperator{\mi}{\mathrm{i}}

\DeclareMathOperator{\argmax}{\text{argmax}}
\DeclareMathOperator{\argmin}{\text{argmin}}

\DeclareMathOperator{\diag}{diag}
\DeclareMathOperator{\ddiag}{ddiag}

\definecolor{xl}{RGB}{200,50,120}

\begin{document}
\title{\bf On the Critical Coupling of the Finite \\
Kuramoto Model on Dense Networks}

\author{Shuyang Ling \thanks{Division of Data Science, New York University Shanghai, Shanghai, China, (Email: sl3635@nyu.edu)} }

\maketitle

\begin{abstract}

Kuramoto model is one of the most prominent models for the synchronization of coupled oscillators. It has long been a research hotspot to understand how natural frequencies, the interaction between oscillators, and network topology determine the onset of synchronization. In this paper, we investigate the critical coupling of Kuramoto oscillators on deterministic dense networks, viewed as a natural generalization from all-to-all networks of identical oscillators. We provide a sufficient condition under which the Kuramoto model with non-identical oscillators has one unique and stable equilibrium. Moreover, this equilibrium is phase cohesive and enjoys local exponential synchronization. We perform numerical simulations of the Kuramoto model on random networks and circulant networks to complement our theoretical analysis and provide insights for future research.

\end{abstract}

\section{Introduction}\label{s:intro}
The synchronization problem of coupled oscillators has a long history, dating back to Christiaan Huygens in 1665 who investigated the behavior of two pendulum clocks mounted side by side on the same support~\cite{ADKM08, RONA16}. Since then, the study of the synchronization phenomenon has attracted a large amount of attention across various scientific areas including mathematics, physics, neuroscience, and engineering. 
Kuramoto model is one classical model of the synchronization of coupled oscillators~\cite{Kura75, Kura84}. Kuramoto considered $n$ fully connected oscillators on a torus whose dynamics is characterized by an ordinary differential equation:

\begin{equation}\label{model:kura1}
\dot{\theta}_i(t) =  \omega_i - \frac{K}{n}\sum_{j=1}^n \sin (\theta_i - \theta_j), \quad 1\leq i\leq n
\end{equation}
where $K$ is called the coupling strength between these oscillators and $\omega_i$ is the natural frequency of the $i$th oscillator. 

One central question about the Kuramoto model is to understand when the oscillators $\{\theta_i(t)\}_{i=1}^n$ will synchronize, the answer of which depends on $K$ and the strength of natural frequencies $\{\omega_i\}_{i=1}^n$. This question has sparked extensive research. 
Moreover, the Kuramoto model has already found numerous applications such as electric power networks, neuroscience, chemical oscillations, spin glasses, see~\cite{ABVRS05, ADKM08, Bullo19, DB12, DB14, DCB13, RPJK16} and the references therein for more details.

In this paper, we consider the Kuramoto model on more general networks, 
\begin{equation}\label{model:kura2}
\dot{\theta}_i(t) =  \omega_i - \frac{K}{n}\sum_{j=1}^na_{ij} \sin (\theta_i - \theta_j), \quad 1\leq i\leq n
\end{equation}
where $\BA = (a_{ij})_{1\leq i,j\leq n}$ is the adjacency matrix of the underlying network connecting these oscillators, i.e.,  two oscillators $i$ and $j$ are linked if and only if $a_{ij}=1$. 
\begin{figure}[h!]
\centering
\includegraphics[width=75mm]{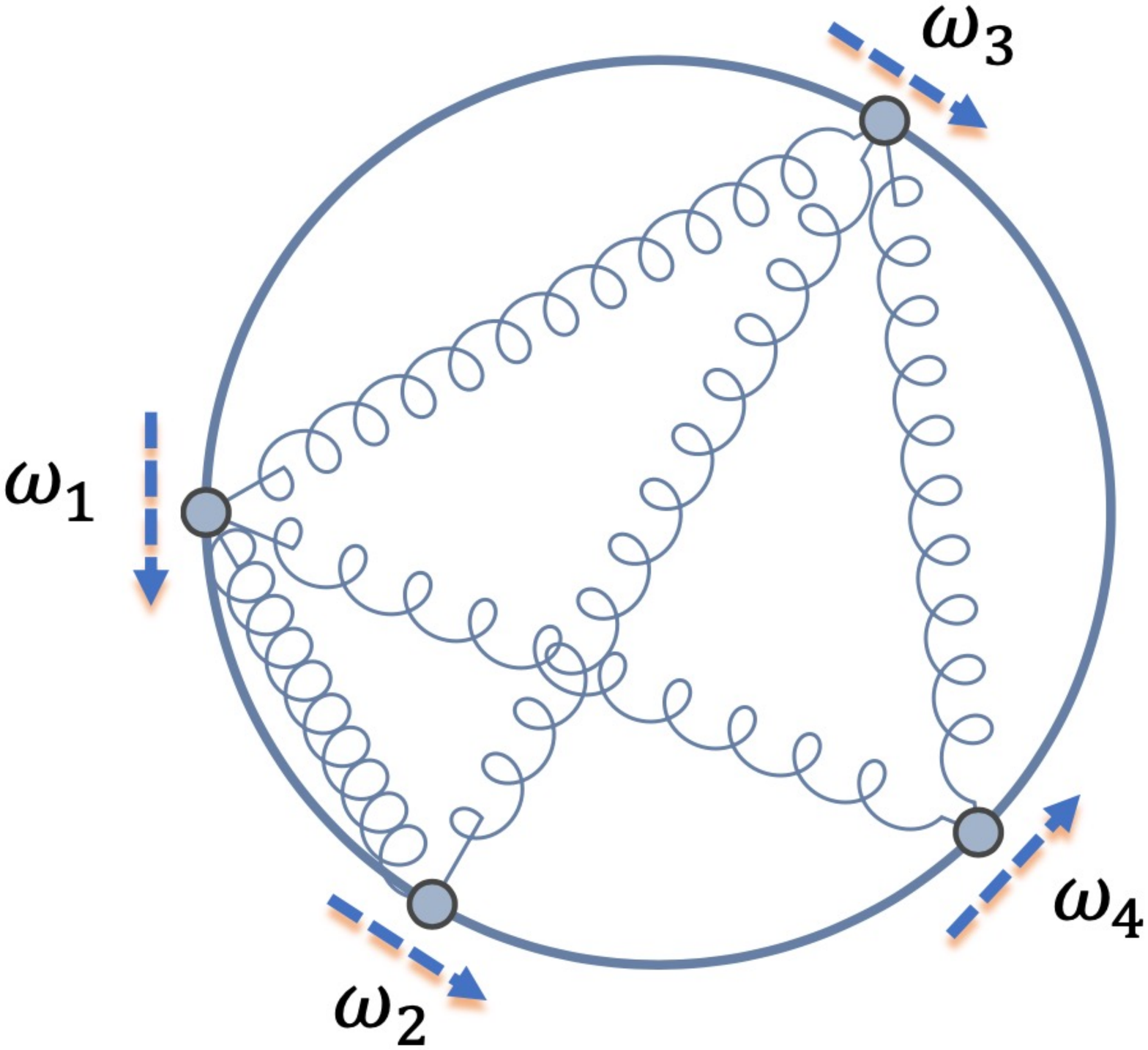}
\caption{The dynamics of Kuramoto model with non-identical oscillators}
\end{figure}
We are interested in the synchronization phenomenon and the long-term behavior of this dynamical system. In particular, we will focus on answering this question:
\begin{equation}\label{question}
\textit{Does there exist a frequency synchronized solution? If so, is it unique?}
\end{equation}
Unlike~\eqref{model:kura1} where the network is complete, the answer to~\eqref{question} will depend on network topology $\BA$ besides the coupling strength $K$ and the dissimilarities of natural frequencies.

\subsection{The state-of-the-art works and our contribution}

As discussed briefly, our work focuses on answering the question~\eqref{question}: the existence and uniqueness of frequency synchronized solution of~\emph{finite} Kuramoto model on~\emph{deterministic dense networks} where each vertex has~\emph{sufficiently many neighbors}. As the synchronization of coupled oscillators is such an important research topic, there inevitably exists an extensive body of literature on it. As a result, we are unable to give an exhaustive literature review here. Instead, we will review the related literature and point out how these previous works inspire ours. 

The study of collective synchronization of coupled oscillators has found quite many applications in biology, physics, and engineering~\cite{Bullo19, MS90, S04, Strogatz00}. Starting from the work by Winfree~\cite{W67}, researchers are interested in developing mathematical models and approaches to analyze and explain the synchronization phenomena. One of the most successful attempts was made by Kuramoto in~\cite{Kura75, Kura84}, who proposed a simplified yet still highly non-trivial model of coupled oscillators~\eqref{model:kura1}. Kuramoto showed that if the coupling $K$ is smaller than a certain threshold $K_c$, then the oscillators are incoherent and transiting to synchrony as $K$  exceeds the threshold $K_c$. This critical coupling threshold is determined by the strength of the natural frequency $\omega_i$ of each oscillator. Since then, a lot of progress has been made to advance our understanding of the Kuramoto model, see~\cite{ABVRS05, ADKM08, Strogatz00} for excellent reviews on this topic. Most of the analyses rely on tools from statistical physics such as calculating the thermodynamical limit of the Kuramoto model. 

As pointed out in~\cite{Strogatz00}, the Kuramoto model with finite oscillators exhibits very different behavior from the scenario where the number of oscillators goes to infinity. It is one major research problem to understand what the critical coupling is for the finite Kuramoto model. In~\cite{DB11}, D\"{o}rfler and Bullo gave a comprehensive and excellent study on the critical coupling of the finite Kuramoto model on a complete graph. In particular,~\cite{DB11} provided a sufficient and necessary condition on the coupling strength for the oscillators to achieve synchronization within the cohesive region where the cohesiveness is later formally defined in~\eqref{def:cohesive}. 
Finite Kuramoto model on general complex networks are discussed in~\cite{ABVRS05, CS09,DB12b, DB14, DCB13,JMB04, RPJK16, VM09}. A review of recent progresses can be found in~\cite{ABVRS05,DB12b,DB14,RPJK16,Strogatz01}.
The work~\cite{JMB04} by Jadbabaie, Motee, and Barahona first estimated the critical coupling of the Kuramoto model on arbitrary complex networks for local convergence and showed the existence and uniqueness of stable fixed point within the cohesive region. Later on, various necessary and sufficient conditions are obtained~\cite{CS09, DB12b, DB14, VM08} for the estimation of the critical coupling.
In particular,~\cite{DB12b, DB14} established sufficient conditions which guarantee frequency synchronization and the existence of locally exponential stable fixed point within the phase-cohesive region. The work~\cite{DCB13} proposed a concise and closed-form condition for synchronization for a large family of networks. These sufficient conditions are formulated by using the second smallest eigenvalue of graph Laplacian associated to the underlying network and the natural frequency $\{\omega_i\}_{i=1}^n.$ The critical coupling of the Kuramoto model is also studied on complete bipartite graphs~\cite{VM09} and random graphs~\cite{CMM18, Ichi04, LLYMG16}.

\vskip0.25cm

It is important to note that Kuramoto model is closely related to the gradient flow of
\begin{equation}\label{def:E}
E(\btheta) := \frac{K}{2n}\sum_{1\leq i,j\leq n} a_{ij}(1 - \cos(\theta_i - \theta_j)) - \sum_{i=1}^n \omega_i \theta_i
\end{equation}
and thus~\eqref{model:kura2} has the following equivalent form:
\[
\frac{\diff \btheta}{\diff t} = -\nabla_{\btheta}E(\btheta),\quad \btheta(0) = \btheta_0
\]
where $\btheta_0$ is the initial state. 
In other words, the Kuramoto oscillators move in the direction which makes $E(\btheta)$ decrease fastest.
Recently,~\cite{LXB19, LS19, Taylor12, TSS19} have studied the global phase synchronization for the homogeneous case $\omega_i = 0$, which is actually a special case of general Kuramoto model. Under $\omega_i = 0$, $E(\btheta)$ in~\eqref{def:E} is a nonnegative energy function and also has been found related to nonconvex optimization approach applied to group synchronization problem on complex networks and community detection~\cite{BBV16,LXB19}. Here the main focus is to understand how the energy landscape of $E(\btheta)$, i.e., the number of stable fixed points of homogeneous Kuramoto dynamics, depends on the network topology. Two main families of networks are investigated: (i): deterministic dense network and circulant networks; (ii): Erd\H{o}s-R\'enyi random graphs~\cite{Hof16}.
The work~\cite{Taylor12} by Taylor first proved that homogenous Kuramoto enjoys global synchronization from almost any initialization if the degree of each node is at least $0.9375n$. Later,~\cite{LXB19} has proven that 0.7929$n$ suffices to guarantee the uniqueness of a globally stable synchronized state, which was recently improved further to 0.7889$n$ in~\cite{LS19} with a more refined analysis. On the other hand,~\cite{CM15, WSG06} have constructed very interesting examples that the uniform twisted states are possibly stable fixed points for homogeneous Kuramoto model on circulant networks. More precisely, there always exists a circulant network whose degree is smaller than $15n/22\approx 0.6818n$ such that the corresponding homogeneous Kuramoto oscillators have a stable equilibrium which is not phase-synchronized. A recent conjecture is proposed in~\cite{TSS19}, stating that the critical threshold could be $0.75n$. The discussion on the global synchronization of homogeneous Kuramoto model has been generalized to other manifolds such as $n$-sphere and Stiefel manifold~\cite{MTG17, MTG20}.
\vskip0.25cm

Our work is motivated by a series of works on the critical coupling of the finite Kuramoto model~\cite{DB12b, DB14, JMB04}, and the recent progress on the uniqueness of stable phase synchronized solution of the homogeneous Kuramoto model on dense networks~\cite{LXB19, LS19, Taylor12, TSS19}. 
Our main contribution is on establishing a sufficient condition under which the inhomogeneous Kuramoto oscillators on dense networks have a unique and stable equilibrium which also is phase-cohesive and satisfies local exponential synchronization. This work is a generalization of the work~\cite{DB11} by D\"{o}rfler and Bullo from all-to-all networks to dense networks. Our work also extends the previous analysis of the energy landscape of the Kuramoto model with identical oscillators~\cite{LXB19, LS19, Taylor12} to the case with heterogeneous oscillators. Moreover, our numerical experiments provide insights for several future directions which are well worth exploring.

\subsection{Organization}
We organize this paper as follows: we will discuss the basics of synchronization and present our main theorem in Section~\ref{s:theorem}. Section~\ref{s:numerics} focuses on numerical simulations and the proof will be given in Section~\ref{s:proofs}. 

\subsection{Notation}
We introduce notations which will be used throughout the paper. Vectors are denoted by boldface lower case letters, e.g.,~$\bz.$
For any vector $\bz$, $\|\bz\|: = \sqrt{\sum_{i=1}^n z_i^2}$ denotes its $\ell_2$-norm; $\|\bz\|_{\infty}:=\max_{1\leq i\leq n}|z_i|$ stands for the $\ell_{\infty}$-norm. For a given matrix $\BZ$, $\ddiag(\BZ)$ denotes a diagonal matrix whose diagonal entries are the same as those of $\BZ$; $\BZ^{\top}$ is the transpose of $\BZ$.
Let $\I_n$ and $\bone_n$ be the $n\times n$ identity matrix and a column vector of ``$1$" in $\RR^n$ respectively. For any symmetric matrix $\BZ\in\RR^{n\times n}$, we denote $\BZ\succeq 0$ if $\BZ$ is positive semidefinite. For any two matrices $\BX$ and $\BZ$ of the same size, their inner product is denoted by $\lag \BX,\BZ\rag := \sum_{i,j}X_{ij}Z_{ij}$ and $\BX\circ \BZ$ denotes their Hadamard product, i.e., $(\BX\circ \BZ)_{ij} = X_{ij}Z_{ij}.$

\section{Preliminaries and main theorem}\label{s:theorem}

We first review the basics of synchronization before moving to our main results.
One can find more discussion on these concepts from many sources such as~\cite[Chapter 16]{Bullo19}.
\begin{definition}[Phase synchronization]
A solution $\btheta(t)\in\RR^n$ achieves phase synchronization if $\theta_i(t) = \theta_j(t)$ for all $t\geq 0.$
\end{definition}
It is well-known that phase synchronization is possible only if $\omega_i = \omega_j$, see~\cite{Bullo19}. Otherwise, $\{\btheta: \theta_i = \theta_j,i\neq j\}$ is not even a fixed point of~\eqref{model:kura1}.
For the Kuramoto model with non-identical oscillators, we focus on frequency synchronization. 
\begin{definition}[Frequency synchronization] 
A solution $\btheta(t)\in\RR^n$ achieves frequency synchronization if $\dot{\theta}_i(t) = \dot{\theta}_j(t)$ for all $t\geq 0.$
\end{definition}

\begin{definition}[Phase cohesiveness]
A solution $\btheta(t)\in\RR^n$ is $\gamma$-phase cohesive if 
there exists a number $\gamma\in[0,\pi)$ such that $\theta(t)\in \Delta(\gamma)$ for all $t\geq 0$ where
\begin{equation}\label{def:cohesive}
\Delta(\gamma) : = \{\btheta: |\theta_i - \theta_j| < \gamma, \quad\forall ~i\neq j\}.
\end{equation}
 In other words, an arc of length $\gamma$ contains all phases $\theta_i(t)$ on the unit circle for any $t\geq 0$.
\end{definition}

Throughout our discussion, we will analyze the oscillators under the rotating frame, i.e., replacing $\omega_i$ by $\omega_i - \frac{1}{n}\sum_{j=1}^n \omega_j$. In this way, we have
\[
\sum_{i=1}^n \dot{\theta}_i(t) = \sum_{i=1}^n \omega_i = 0.
\]
In other words, a frequency-synchronized solution $\btheta(t)$ is a fixed point of~\eqref{model:kura2}, 
\begin{equation}\label{def:fp}
\dot{\theta}_i(t) = 0 \Longleftrightarrow \frac{K}{n}\sum_{j=1}^na_{ij} \sin(\theta_i - \theta_j) = \omega_i, \quad 1\leq i\leq n.
\end{equation}
The stability is determined by the spectra of Jacobian matrix. The Jacobian matrix is equal to $-n^{-1}K\cdot \BJ(\btheta)$ where 
\begin{equation}\label{def:J}
\BJ(\btheta) := \ddiag(\BA\BQ\BQ^{\top}) - \BA\circ \BQ\BQ^{\top}.
\end{equation}
Here $\BA = (a_{ij})_{1\leq i,j\leq n}$ is the adjacency matrix, $\BQ\in\RR^{n\times 2}$ with its $i$th row $\bq_i \in [\cos(\theta_i), \sin(\theta_i)]$, and $(\BA\circ\BQ\BQ^{\top})_{ij} := a_{ij}\cos(\theta_i - \theta_j)$ is the Hadamard product of $\BA$ and $\BQ\BQ^{\top}$. The $(i,j)$-entry of $\BJ(\btheta)$ is
\[
(\BJ(\btheta))_{ij} = 
\begin{cases}
\sum_{j=1}^n a_{ij} \cos(\theta_i - \theta_j),& i = j,\\
-a_{ij}\cos(\theta_i - \theta_j), & i \neq j, 
\end{cases}
\]
which is actually a Laplacian matrix associated with the signed weights $\{a_{ij}\cos(\theta_i-\theta_j)\}_{i,j}$.
From the construction, we know that $\btheta(t)$ is a stable frequency synchronized solution if $\|\dot{\btheta}\| = 0$ and the second smallest eigenvalue $\lambda_2(\BJ(\btheta)) > 0$ since 0 is always an eigenvalue of $\BJ(\btheta).$

As discussed before, a complete understanding of synchronization of the Kuramoto model on general complex networks is still one major open problem in the field. Our discussion will mainly focus on deterministic dense networks, which are defined formally as follows.
\begin{definition}[Deterministic dense networks]
A network with symmetric adjacency matrix $\BA = (a_{ij})_{1\leq i,j\leq n} \in \{0,1\}^{n\times n}$ is $\mu$-dense if the degree of each node is at least $\mu(n-1)$, $0\leq \mu\leq 1$. In particular, if $\mu = 1$, the corresponding network is complete.
\end{definition}
The Kuramoto model on  dense networks is a natural generalization of the classical Kuramoto model with all-to-all coupling. Now we return to the main questions~\eqref{question}:
\begin{enumerate}[(a)]
\item When do the coupled oscillators on $\mu$-dense network have a stable frequency synchronized solution? 
\item If such a solution exists, is it unique? 
\end{enumerate}
The answer depends on the parameter $\mu$, critical coupling parameter $K$, as well as the strength of natural frequency $\omega_i.$ This will be the main focus of our paper. 




\begin{theorem}\label{thm:main}

Consider the Kuramoto model~\eqref{model:kura2} with natural frequency $\{\omega_i\}_{i=1}^n$ and $\sum_{i=1}^n \omega_i = 0$. Suppose
\begin{equation}\label{eq:main}
\|\omega\|_{\infty} : = \max_{1\leq i\leq n} |\omega_i| < K\left( \sqrt{\mu - \frac{3}{4}} + \mu - 1\right),
\end{equation}
then there exists a unique frequency-synchronized solution.
Moreover, this solution is located within the phase cohesive region $\Delta(\pi/2)$ and the coupled oscillators achieve local exponential frequency synchronization.

\end{theorem}

\begin{remark}
From~\eqref{eq:main}, we see that $\mu \geq \frac{3-\sqrt{2}}{2}$ is implicitly assumed since the right hand side of~\eqref{eq:main} is nonnegative. Therefore, this network is always connected.
\end{remark}

We briefly discuss what this theorem implies by looking at two interesting special cases: (i) $\mu=1$; (ii) $\|\omega\|_{\infty} = 0$, and make comparisons with the state-of-the-art results. 

\begin{enumerate}[(i)]

\item Heterogeneous oscillators over a complete graph $\mu=1$:

The model~\eqref{model:kura2} immediately reduces to the classical Kuramoto model~\eqref{model:kura1}. Then the condition~\eqref{eq:main} reads
$K > 2\|\omega\|_{\infty}$. Note that~\cite[Theorem 4.1]{DB11} shows that $K \geq \omega_{\max} - \omega_{\min}$ suffices to guarantee a local exponentially stable frequency synchronized solution within the phase cohesive region $\Delta(\gamma)$ for some $\gamma < \pi/2$. Compared with~\cite{DB11}, our result is slightly looser and becomes tight if $\omega_{\max} + \omega_{\min} = 0.$ 

\item Homogeneous Kuramoto model on arbitrary unweighted networks:

If $\omega_i = 0$ holds for all $1\leq i\leq n$, then~\eqref{eq:main} turns into
\[
\sqrt{\mu - \frac{3}{4}} + \mu-1 > 0 ~~\Longleftrightarrow ~~\mu > \frac{3-\sqrt{2}}{2} \approx 0.7929.
\]
In other words, the stable equilibrium is unique which is exactly the phase synchronized solution, i.e., $\theta_i = \theta_j$, if $\mu > 0.7929$. This matches the author's previous work~\cite{LXB19}. This result is later improved to $\mu > 0.7889$ in~\cite{LS19}. On the other hand, for any $\mu< 15/22\approx 0.6818$, there always exists a $\mu$-dense circulant network such that the corresponding Kuramoto model has multiple stable equilibria, i.e., phase/frequency synchronized solutions. In other words, for $\mu < 15/22$, there always exist multiple stable frequency synchronized solutions if $\omega_i$ is sufficiently small. Adding tiny  natural frequencies to the homogeneous Kuramoto model, regarded as a small perturbation, will not change the number of stable solutions. We provide a numerical illustration in Figure~\ref{fig:WSG_twist}.

\end{enumerate}

\section{Numerics}\label{s:numerics}
One natural question is the tightness of the bound~\eqref{eq:main} in Theorem~\ref{thm:main} as well as the possible extension of this result to other families of networks.
We address this issue by providing numerical examples. In particular, we will focus on two types of graphs: Erd\H{o}s-R\'enyi random graphs and circulant networks.

\subsection{Frequency synchronization on Erd\H{o}s-R\'enyi random graphs}
Note that the condition in~\eqref{eq:main} only holds for very dense networks,  viewed as the worst-case scenario. It is natural to ask what the critical coupling is for Erd\H{o}s-R\'enyi (ER) random graphs. We denote ER graph as ${\cal G}(n,p)$ if the network is of size $n$ and its $(i,j)$-entry $a_{ij}$ of the symmetric adjacency matrix is a Bernoulli random variable with parameter $p.$

We want to study how the frequency synchronization depends on $p$ and the strength of natural frequency. 
We let network contain $n=100$ vertices and $K=1$, and sample $\BA$ from ${\cal G}(n,p)$ with $p$ varying from 0.05 to 1. For each random instance, we generate uniformly distributed $\omega_i$ over $[-\kappa p, \kappa p]$ for $\kappa = 0,0.05,\cdots,0.95,1$. In other words, the range of $\omega_i$ is approximately proportional to the expected degree $np$ of each node.
For each network, we use Euler scheme to simulate the trajectory with step size $\tau = 2\cdot10^{-2}.$ The simulation stops if the iterate converges to an approximate stable frequency-synchronized solution, i.e.,
\[
\|\dot{\btheta}\| < 10^{-4}, \quad \lambda_{2}(\BJ(\btheta)) \geq 10^{-8},
\]
or the iteration number reaches $10^4$. We run 20 experiments for each pair of $(p, \kappa)$, and calculate the proportion of the iterate converging to a stable synchronized solution. 

The phase transition plot is in Figure~\ref{fig:ER}: the white region stands for success and black means failure. We can see if $\kappa = \|\omega\|_{\infty}/p < 0.75$ and $p > 0.2$, the dynamics will finally synchronize with high probability. Empirically, it seems $\kappa$ does not heavily depend on $p$. We also compute the length of the shortest arc covering all $\theta_i$.
Figure~\ref{fig:ER} indicates that if $\kappa < 0.7$, i.e., $\|\omega\|_{\infty} \leq 0.7p$, the oscillators will converge to frequency-synchronized solution in the cohesive region $\Delta(\pi/2).$
In fact, Figure~\ref{fig:ER} matches our previous numerical study in~\cite{LXB19}. For the homogeneous Kuramoto model $\kappa = 0$ ($\omega_i = 0$), the oscillators always achieve global synchronization for $p > n^{-1}\log n$, i.e., as long as the ER graph is connected with high probability, there exists a synchronized solution. 

Compared with the bound in~\eqref{eq:main}, it shows that our existing bound is relatively conservative since it does not assume an additional probabilistic structure on the networks. We believe one can obtain a refined bound on the critical coupling threshold by taking this extra prior information into account.

\begin{figure}[h!]
\begin{minipage}{0.48\textwidth}
\centering
\includegraphics[width=82mm]{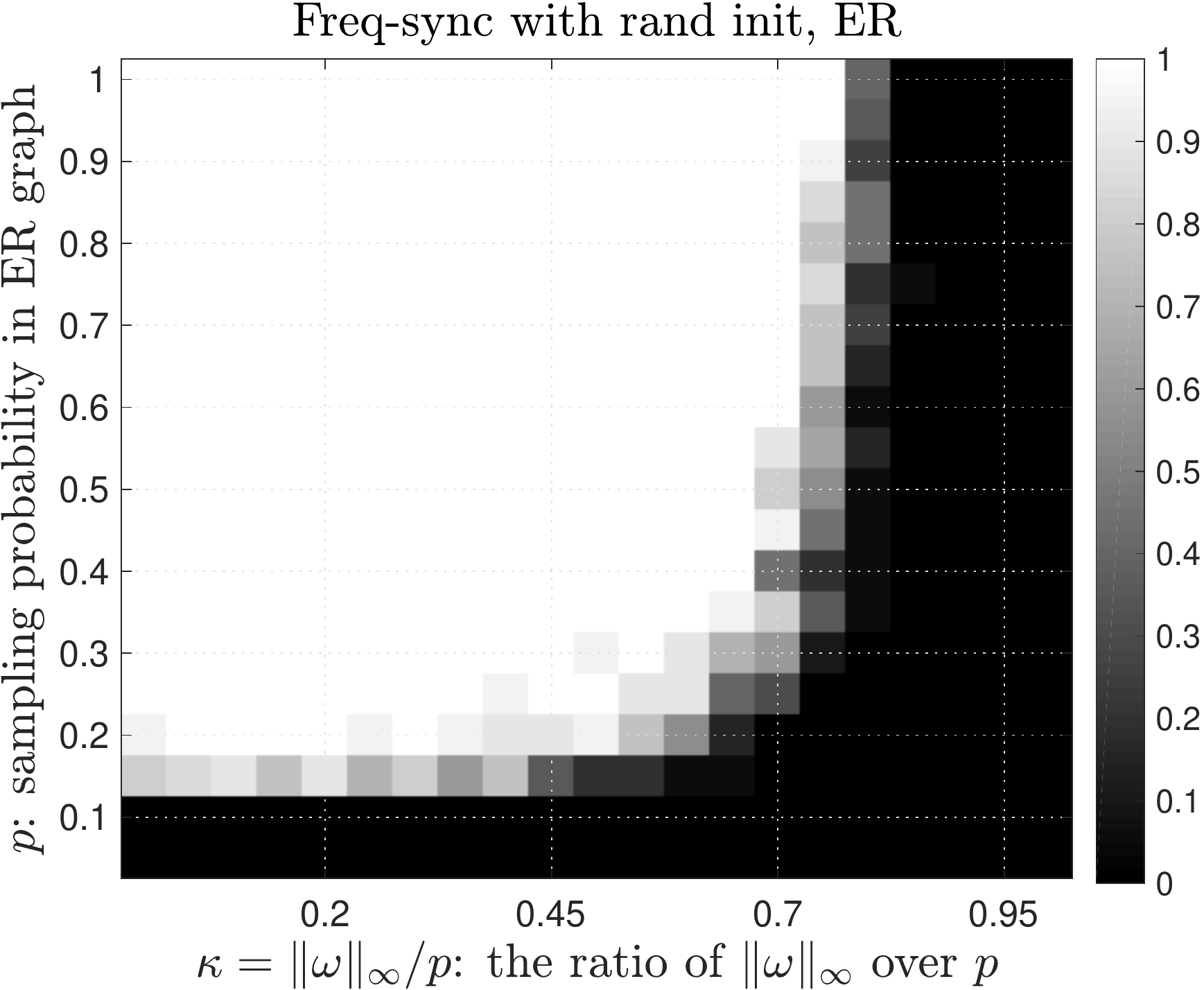}
\end{minipage}
\hfill
\begin{minipage}{0.48\textwidth}
\centering
\includegraphics[width=82mm]{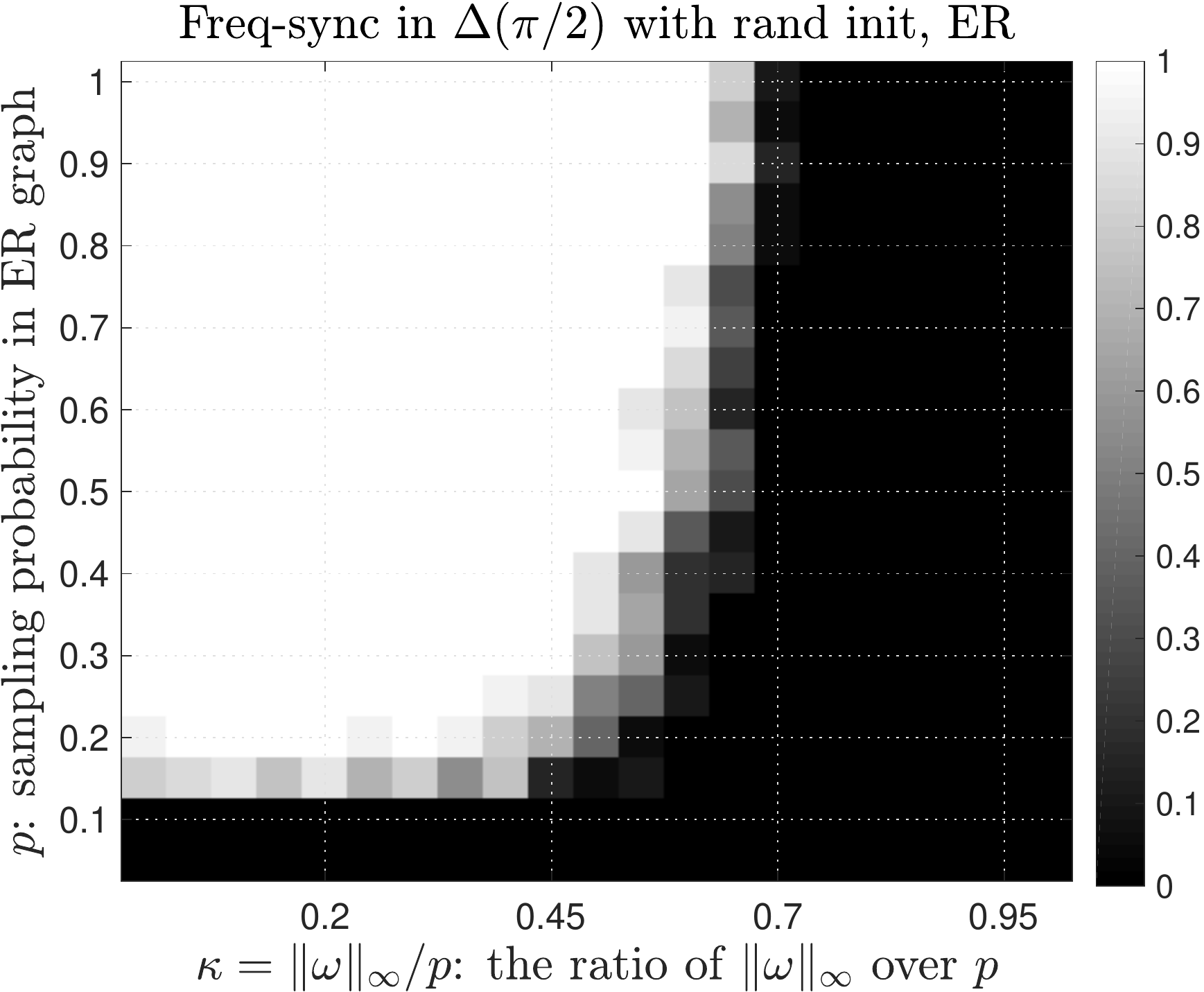}
\end{minipage}
\caption{Phase transition plot for frequency synchronization solutions: ER graphs and random initialization.}
\label{fig:ER}
\end{figure}

\subsection{Frequency synchronization on circulant networks}

The homogeneous Kuramoto model on circulant networks exhibits very interesting behaviors. It has been shown that the uniform twisted state $\btheta_{\text{twist}} := 2\pi n^{-1}[1,\cdots,n]$ is a stable equilibrium for a large family of circulant networks such as Wiley-Strogatz-Girvan (WSG) networks~\cite{WSG06} and Harary graphs~\cite{CM15}. Now we will focus on the critical coupling thresholds on WSG networks. 

The WSG$(k)$ networks are constructed by starting with an $n$-cycle graph and then connecting each node with its closest $k$ neighbors. As a result, the degree of each node is $2k.$
We set up numerical experiments as follows: let $k = \lfloor np/2\rfloor$, i.e., the largest integer smaller than or equal to $np/2$, for $ p =0.05,0.1,\cdots,1$. Note that WSG($k$) is also a $\mu$-dense network with $\mu = p.$
We generate the natural frequencies $\omega_i$ according to uniform distribution $[-\kappa p,\kappa p]$ for $\kappa$ ranging from $0,0.05, \cdots,1.$ For each pair of $(\kappa, p)$, we run 20 experiments and count the number of cases in which the iterates converge to a stable synchronized solution. 
In particular, to address the dependence of the dynamics on initialization, two types of initial states are chosen: (i) $\btheta(0)$ is uniformly distributed over $[0,2\pi]$; (ii) $\btheta(0)=\frac{2\pi}{n}[1,\cdots,n]$ is a uniform twisted state.

From Figure~\ref{fig:CN_rand}, we can see that the phase transition plot looks similar to Figure~\ref{fig:ER} for random initialization but the white region shrinks: on circulant networks, random initialization will lead to a frequency-synchronized solution if $p \geq 0.4$ and $\|\omega\|_{\infty} \leq 0.7 p$ while $p \geq 0.2$ and $\|\omega\|_{\infty} \leq 0.75p$ suffice for ER graphs. 

However, if the initialization is~\emph{a uniform-twisted state}, Figure~\ref{fig:CN_unif} looks very different from Figure~\ref{fig:CN_rand}. 
There exists a large set of $(\kappa,p)$ in Figure~\ref{fig:CN_unif} such that starting from the uniform twisted state, the trajectory does not converge to a cohesive solution. We present a more concrete example in Figure~\ref{fig:WSG_twist}: Figure~\ref{fig:WSG_twist} demonstrates that the oscillators reach a stable non-cohesive frequency-synchronized solution outside $\Delta(\pi/2)$ for small $\omega_i$, with initialization $\btheta_{\text{twist}}$. However, this stable non-cohesive solution disappears if the intrinsic frequency gets stronger. Instead, they converge to a stable solution within $\Delta(\pi/2).$

\begin{figure}[h!]
\begin{minipage}{0.48\textwidth}
\centering
\includegraphics[width=82mm]{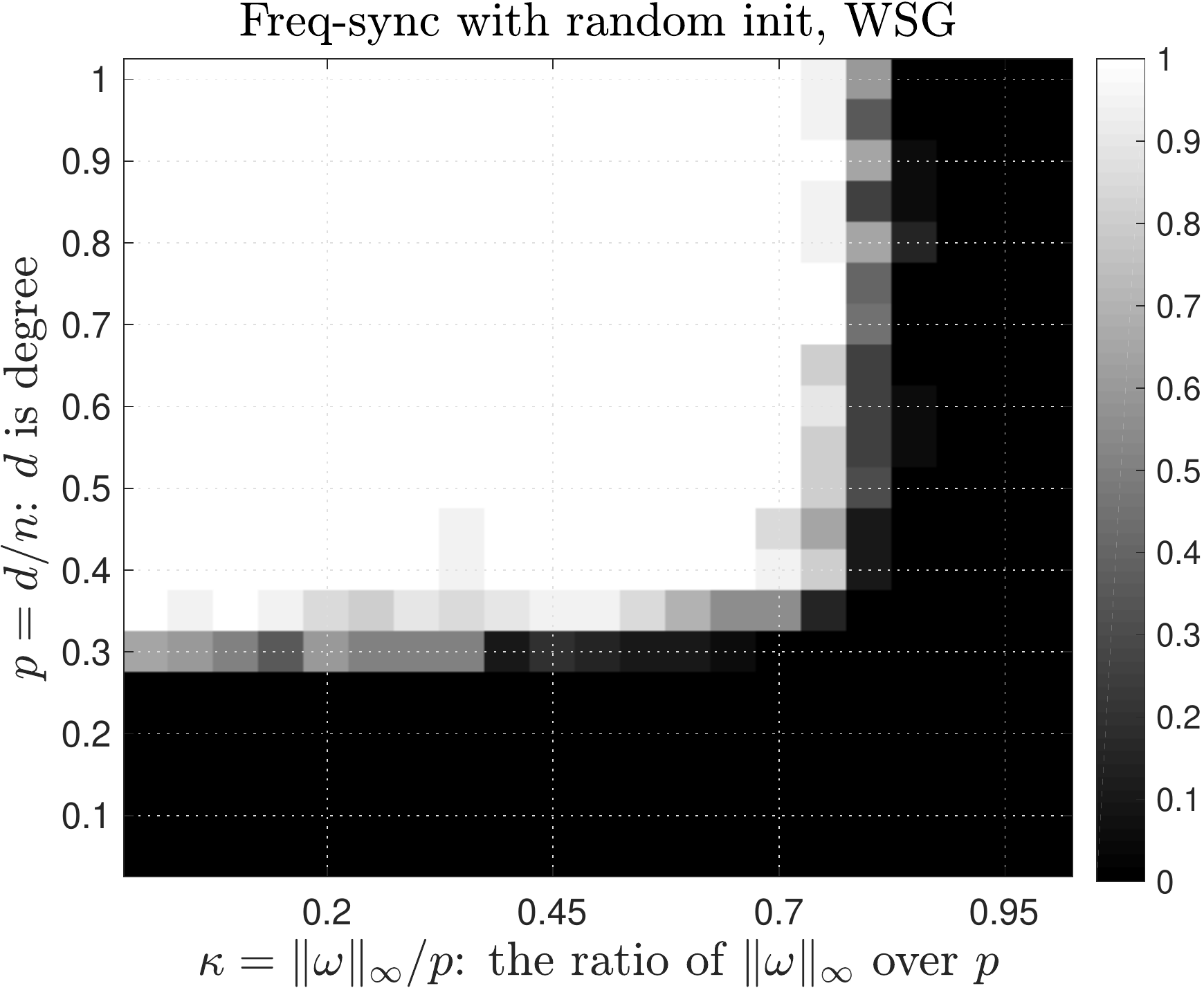}
\end{minipage}
\hfill
\begin{minipage}{0.48\textwidth}
\centering
\includegraphics[width=82mm]{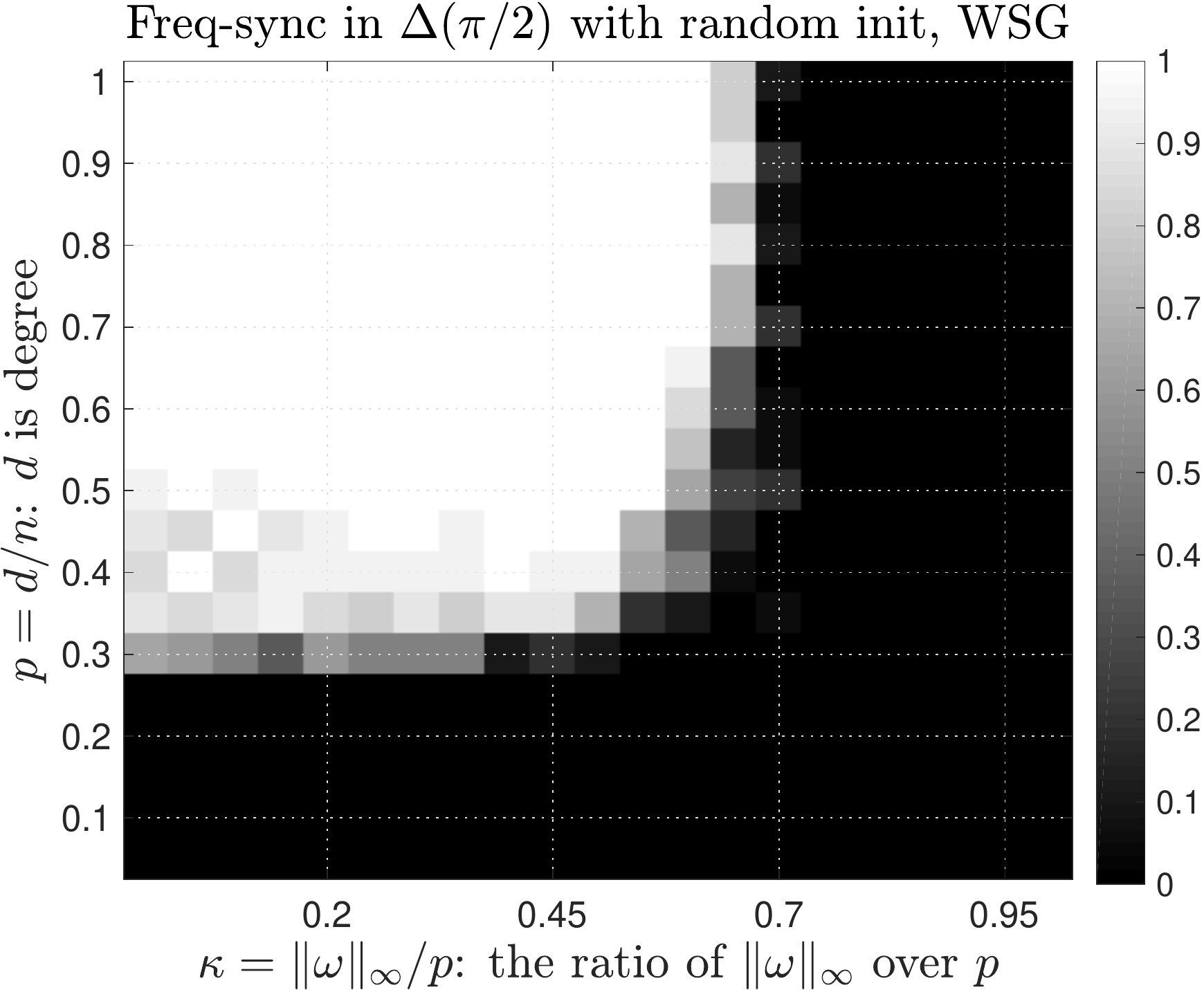}
\end{minipage}
\caption{Phase transition plot for frequency synchronization solutions: WSG circulant networks with degree $np$ and uniformly random initialization.}
\label{fig:CN_rand}
\end{figure}

Here is one explanation for Figure~\ref{fig:CN_unif}: if $\omega_i = 0$, the uniform-twisted state is always a stable equilibrium (with strictly positive second smallest eigenvalue of $\BJ(\btheta)$) if $k$ is smaller than 0.34 in~\cite{WSG06}. If $\omega_i$ is small compared with the second smallest eigenvalue $\lambda_2(\BJ(\btheta))$, this perturbation on $\omega_i$ does not change dynamics qualitatively. The uniform-twisted state still lies in the basin of attraction and the oscillators will converge to a non-cohesive frequency synchronized solution if starting from $\btheta_{\text{twist}}$. 
On the other hand, if the intrinsic frequency gets larger, the iterate will escape from the basin of attraction around the uniform twisted solution and reach another synchronization solution if there exists one.
If $\omega_i$ gets too strong, a stable synchronized state no longer exists. 

That is why the left plot in Figure~\ref{fig:CN_unif} shows that for fixed $p < 0.6$, the oscillators synchronize to a~\emph{non-cohesive} solution for small $\omega_i$; as $\omega_i$ gets stronger, it starts to become incoherent; if the strength of $\omega_i$ exceeds certain threshold, the oscillators start synchronizing to a~\emph{cohesive} synchronized solution; but when $\|\bomega\|_{\infty}$ is larger than $0.7p$, the system no longer synchronizes. The two plots in Figure~\ref{fig:CN_unif} indicate that there exists a~\emph{boundary} for $(\kappa,p)$ which distinguishes a cohesive synchronized state from to a non-cohesive one. We leave the explanation of this boundary for future research. On the other hand, if $p > 0.7$ and the network is very densely connected, either choosing a random initialization or picking $\btheta(0) = \btheta_{\text{twist}}$ yields a phase-cohesive synchronized solution for $\|\bomega\|_{\infty} < 0.7p.$ It is because $\btheta_{\text{twist}}$ is not a stable equilibrium for the homogeneous Kuramoto model if $p > 0.68$; also the interaction between oscillators gets stronger than the effect of natural frequency due to the higher network connectivity, which makes the global synchronization more likely.

\begin{figure}[h!]
\centering
\begin{minipage}{0.48\textwidth}
\centering
\includegraphics[width=80mm]{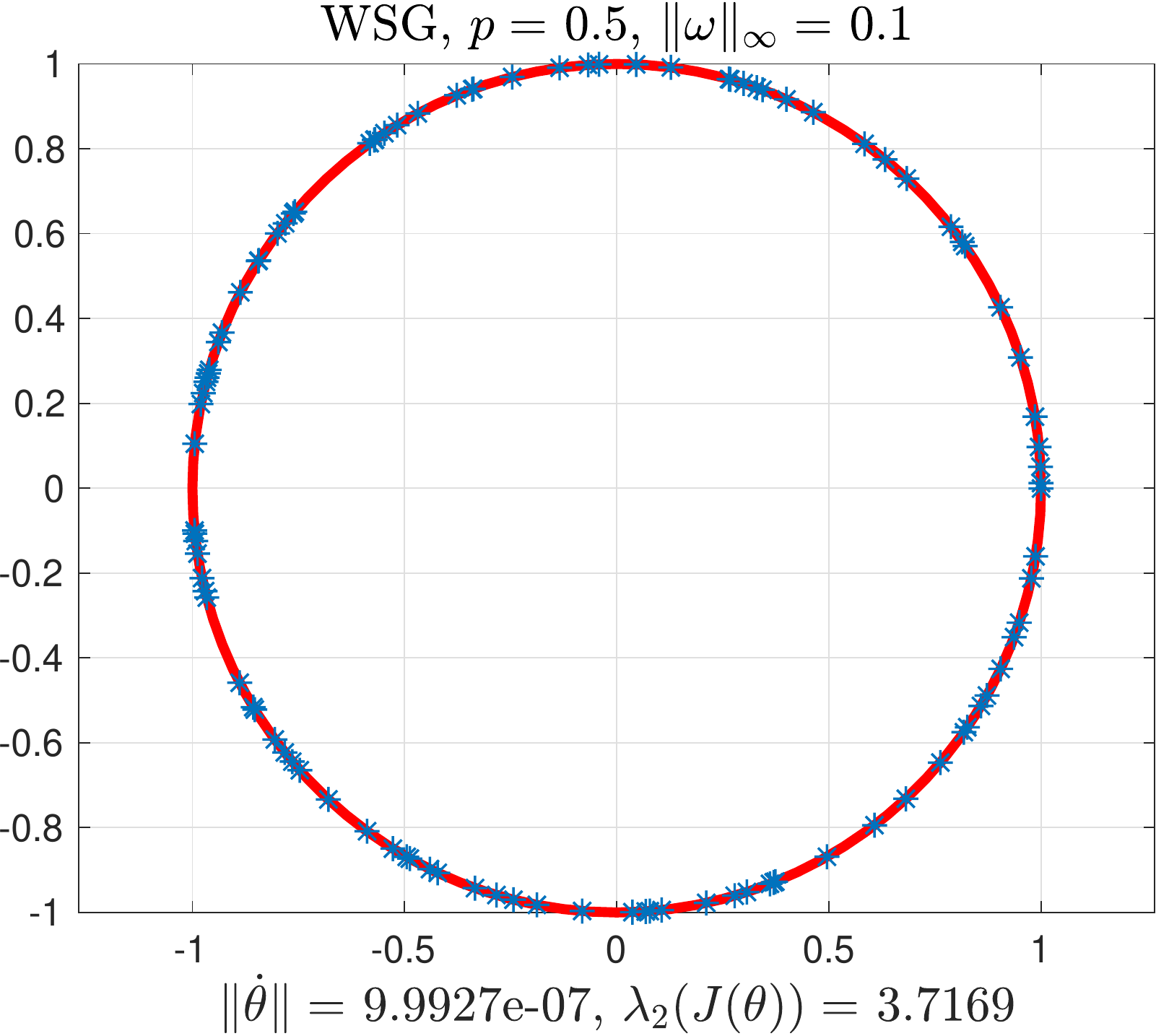}
\end{minipage}
\hfill
\begin{minipage}{0.48\textwidth}
\centering
\includegraphics[width=80mm]{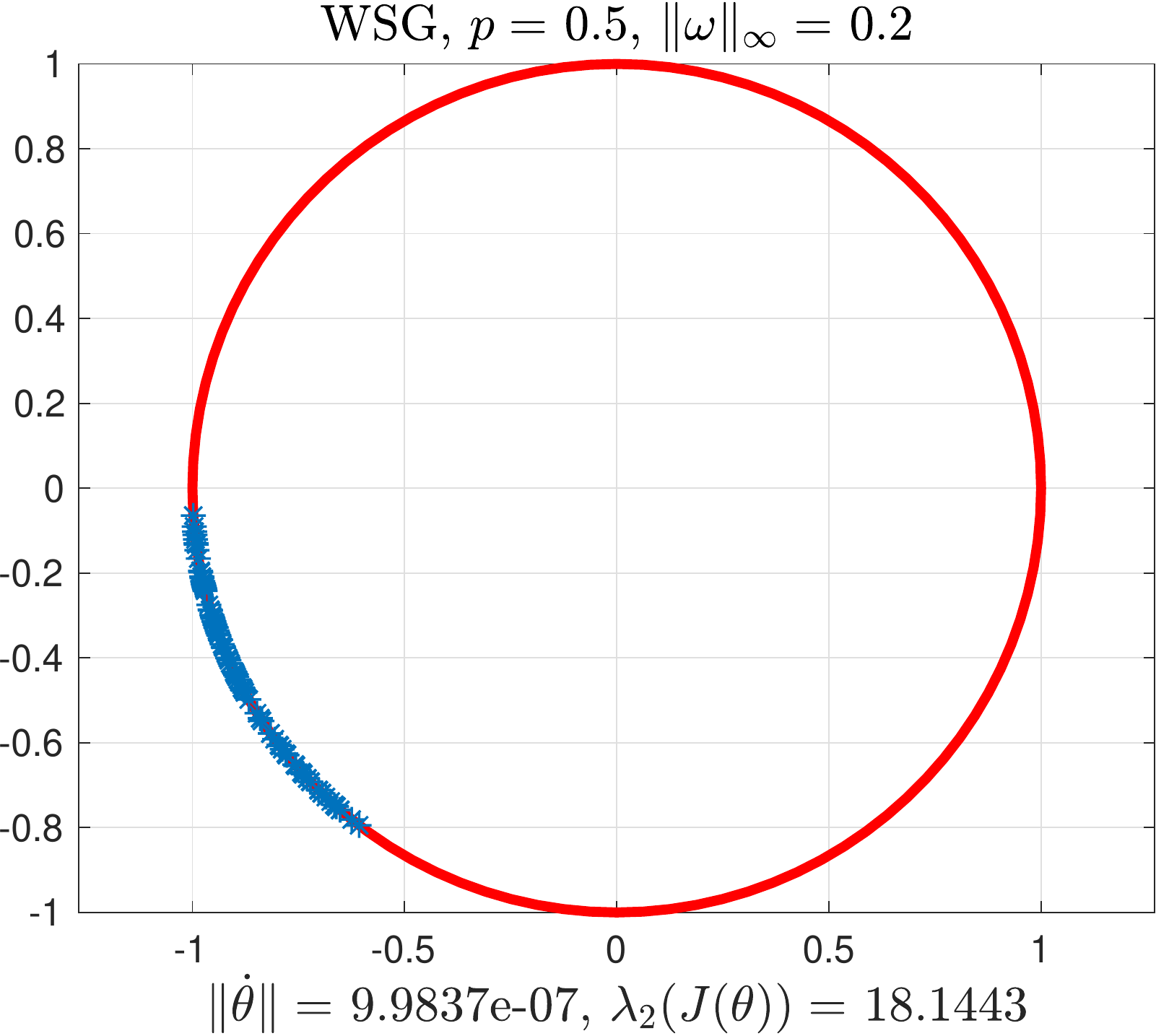}
\end{minipage}
\caption{Phase cohesiveness. Left: $p=0.5$ and $\|\omega\|_{\infty} = 0.1$. In this case, a non-cohesive synchronized state exists for small $\omega_i$ ; Right: $p=0.5$ and $\|\omega\|_{\infty} = 0.2$; For large $\omega_i$, setting $\btheta(0)=\btheta_{\text{twist}}$ still leads to a non-cohesive stable equilibrium.}
\label{fig:WSG_twist}
\end{figure}

\begin{figure}[h!]
\begin{minipage}{0.48\textwidth}
\centering
\includegraphics[width=80mm]{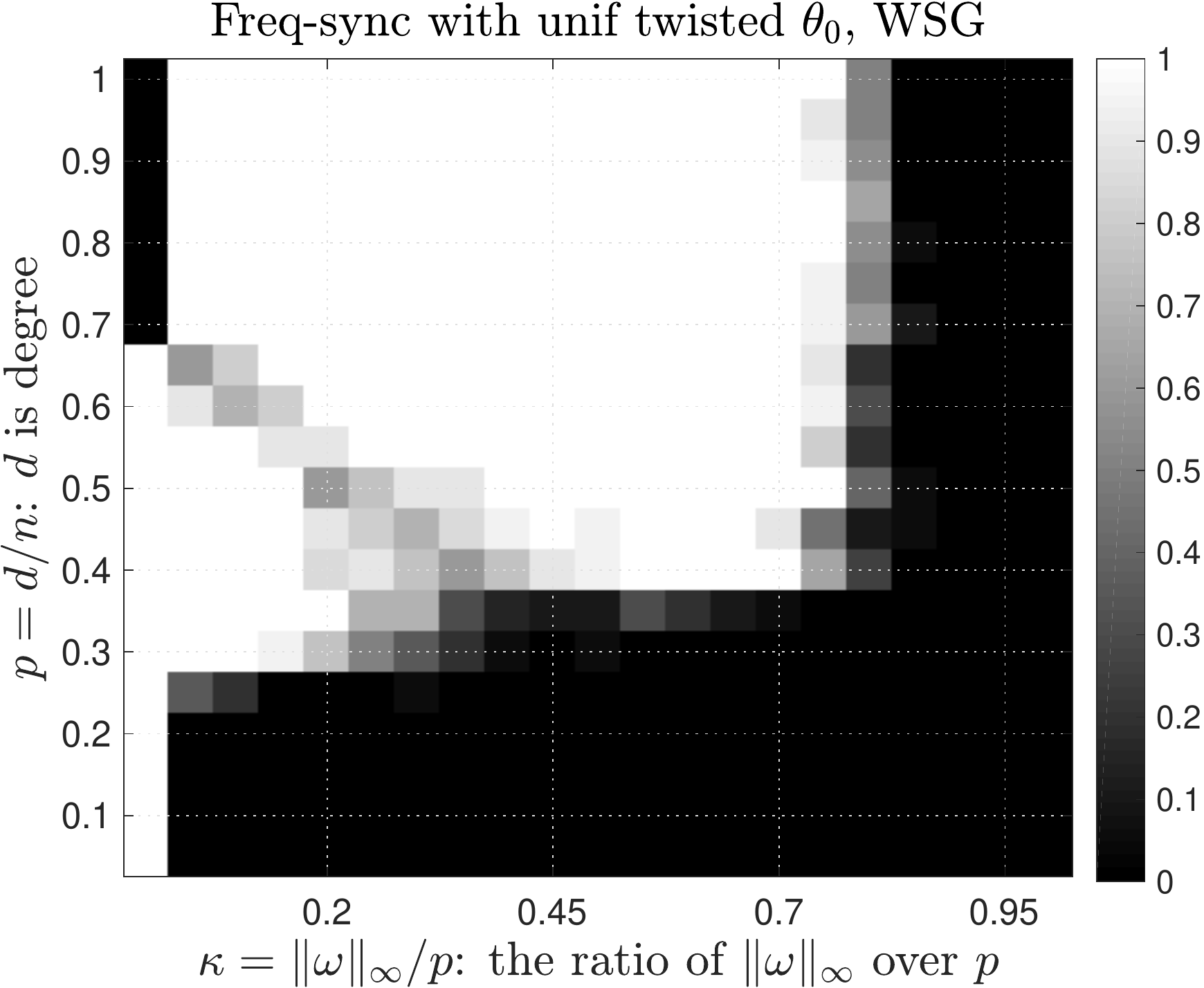}
\end{minipage}
\hfill
\begin{minipage}{0.48\textwidth}
\centering
\includegraphics[width=80mm]{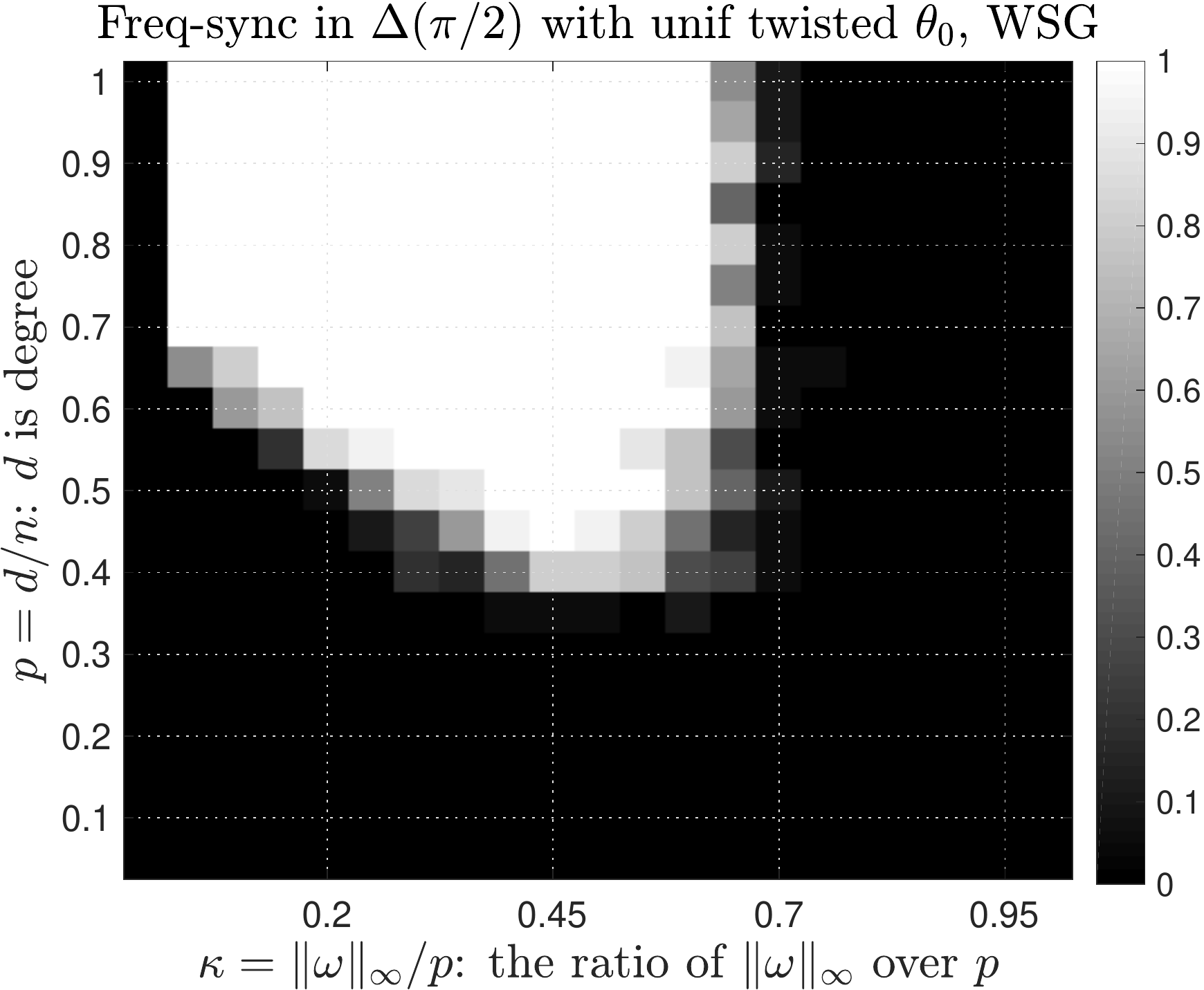}
\end{minipage}
\caption{Phase transition plot for frequency synchronization solutions: WSG circulant networks with degree $np$ and initialization $\btheta(0) = \btheta_{\text{twist}}.$} 
\label{fig:CN_unif}
\end{figure}

\section{Proofs}\label{s:proofs}

\subsection{Characterization of stable equilibria}

Define the unnormalized order parameter
\[
r(\btheta) = \sum_{j=1}^n e^{\mi \theta_j}
\]
whose magnitude equals 1 if all oscillators are fully synchronized.

\begin{lemma}\label{lem:quad}
Consider the Kuramoto model~\eqref{model:kura2} with natural frequency $\{\omega_i\}_{i=1}^n$ and $\sum_{i=1}^n \omega_i = 0$. Suppose~\eqref{eq:main} holds,
then all the stable equilibria, if existing, must be in a cohesive region $\Delta(\pi/2).$
\end{lemma}

\begin{proof}[Proof of Lemma~\ref{lem:quad}]
The proof generalizes the idea in~\cite[Theorem 3.1]{LXB19}: first show that if an equilibrium is stable, its order parameter is sufficiently large, which is made possible by using the second-order necessary condition; then the first-order critical condition guarantees that all oscillators are inside the cohesive region $\Delta(\gamma)$ for some $\gamma < \pi/2$.

Suppose $\btheta$ is a stable equilibrium, then the Jacobian matrix is negative definite and thus $\BJ(\btheta) = \diag(\BA\BQ\BQ^{\top}) - \BA\circ \BQ\BQ^{\top} \succeq 0$ holds where $\BJ(\btheta)$ is defined in~\eqref{def:J}. As a result, 
\begin{equation}\label{eq:2nd}
\lag \BJ(\btheta), \BQ\BQ^{\top}\rag \geq 0 \Longleftrightarrow \lag  \BA, \BQ\BQ^{\top}\rag \geq \lag \BA, \BQ\BQ^{\top}\circ \BQ\BQ^{\top}\rag
\end{equation}
since $\BQ\BQ^{\top}\succeq 0$ with diagonal entries equal to 1. This condition immediately implies a lower bound of the magnitude of the unnormalized order parameter:
\begin{align}
\|r(\btheta)\|^2 & = \sum_{i=1}^n\sum_{j=1}^n\cos(\theta_i - \theta_j)  = \lag \bone_n\bone_n^{\top}, \BQ\BQ^{\top}\rag \\
&  = \lag \BA, \BQ\BQ^{\top}\rag + \lag \bone_n\bone_n^{\top} - \BA, \BQ\BQ^{\top} \rag \nonumber\\
& \geq  \lag \BA,\BQ\BQ^{\top}\circ \BQ\BQ^{\top}\rag + \lag \bone_n\bone_n^{\top} - \BA, \BQ\BQ^{\top} \rag \nonumber \\
& = \lag \bone_n\bone_n^{\top}, \BQ\BQ^{\top}\circ \BQ\BQ^{\top}\rag + \lag \bone_n\bone_n^{\top} - \BA, \BQ\BQ^{\top}-\BQ\BQ^{\top}\circ \BQ\BQ^{\top}\rag \label{eq:r2}
\end{align}
where the first inequality follows from~\eqref{eq:2nd} and $\bone_n$ is the $n\times 1$ vector with all entries equal to 1.
Note that $\BQ\BQ^{\top}$ is rank-2, positive semidefinite, and has its trace equal to $n$. Therefore, 
\[
\lag \bone_n\bone_n^{\top}, \BQ\BQ^{\top}\circ \BQ\BQ^{\top}\rag = \|\BQ\BQ^{\top}\|_F^2 \geq \frac{1}{2} |\Tr(\BQ\BQ^{\top})|^2 = \frac{n^2}{2}
\]
where $\|\cdot\|_F$ stands for the matrix Frobenius norm.
For the second term in~\eqref{eq:r2}, we simply use the fact that $|(\BQ\BQ^{\top})_{ij}| \leq 1$ and $\lag \BA, \bone_n\bone_n^{\top}\rag$ is at least $\mu n(n-1) + n$ since $\BA$ is $\mu$-dense. Then
\begin{align*}
& \lag \bone_n\bone_n^{\top} - \BA, \BQ\BQ^{\top}-\BQ\BQ^{\top}\circ \BQ\BQ^{\top} \rag  \geq -2 \lag \bone_n\bone_n^{\top} - \BA, \bone_n\bone_n^{\top}\rag \\
& \qquad\geq  -2n^2 +2\mu n(n-1) + 2n = -2 (1- \mu )n^2 + 2(1-\mu)n.
\end{align*}
Therefore, we get
\begin{equation}\label{eq:r_low}
\|r(\btheta)\|^2 \geq \frac{n^2}{2} -2 (1- \mu )n^2 + 2(1-\mu)n = \left( 2\mu - \frac{3}{2}\right)n^2 + 2(1-\mu)n.
\end{equation}

Next, we bound the phase of the fixed point (if it exists).
Let's compute the phase difference between $r(\btheta)$ and $i$th oscillator $e^{\mi\theta_i}$
\[
re^{-\mi\theta_i} = \sum_{j=1}^n e^{\mi (\theta_j - \theta_i)}= \sum_{j=1}^n \left(\cos (\theta_j - \theta_i) + \mi \sin(\theta_j - \theta_i) \right)= \|r\| \cdot e^{\mi(\psi - \theta_i)} 
\]
where $\psi$ is the phase of $r(\btheta).$
By taking the imaginary part, it holds that
\begin{align*}
\|r\|\cdot \sin(\psi - \theta_i) & =  \sum_{j=1}^n \sin(\theta_j - \theta_i)  \\
& =   \sum_{j=1}^n a_{ij}\sin(\theta_i - \theta_j) + \sum_{j=1}^n (1-a_{ij})\sin(\theta_i - \theta_j) \\
& = nK^{-1}\omega_i + \sum_{j=1}^n (1-a_{ij})\sin(\theta_i - \theta_j) 
\end{align*}
where the last inequality uses $\dot{\theta}_i = 0$ and $\sum_{j=1}^n a_{ij}\sin(\theta_i - \theta_j) = nK^{-1}\omega_i.$
Therefore, we obtain 
\begin{equation}\label{eq:1st}
\|r\|\cdot |\sin(\psi - \theta_i)| \leq nK^{-1}\|\omega\|_{\infty} + (1 - \mu)(n-1)
\end{equation}
where the right hand side is independent of the index $i$. 
Combining~\eqref{eq:1st} with~\eqref{eq:r_low} leads to an upper bound of $|\sin(\psi - \theta_i)|$ and thus $|\psi - \theta_i|$, 
\begin{equation}\label{eq:sin}
 |\sin(\theta_i -\psi)| \leq \frac{nK^{-1}\|\omega\|_{\infty} + (1-\mu)(n-1)}{\sqrt{\left(2\mu - \frac{3}{2}\right)n^2 + 2(1-\mu)n}},~~1\leq i\leq n.
\end{equation}

\vskip0.2cm

We claim that under~\eqref{eq:main}, $\{\theta_i\}_{i=1}^n$ are inside two disjoint quadrants.
\begin{equation}\label{eq:quad}
\theta_i \in \left\{ \theta:  \psi - \frac{\pi}{4} < \theta < \psi + \frac{\pi}{4} \right\}\cup \left\{ \theta:  \psi + \frac{3\pi}{4} < \theta < \psi + \frac{5\pi}{4} \right\}.
\end{equation}
To ensure this, it suffices to bound~\eqref{eq:sin} by $1/\sqrt{2}$ which is guaranteed by
\[
\frac{\|\omega\|_{\infty}}{K}  \leq \sqrt{\mu - \frac{3}{4} + \frac{1-\mu}{n}}  + \mu - 1.
\]

\vskip0.25cm

In fact, as will be shown below, all $\{\theta_i\}_{i=1}^n$ are in the same quadrant if $\btheta$ is a stable equilibrium. Otherwise, we can construct a test vector such that the quadratic form w.r.t. $\BJ(\btheta)$ is smaller than 0, as discussed in~\cite{Taylor12,LXB19}.
We provide the proof here to make the presentation self-contained. Remember the stability is directly to the second smallest eigenvalue of $\BJ(\btheta)$ since the eigenvalues of Jacobian matrix and $\BJ(\btheta)$ have completely opposite signs. Suppose not every $\theta_i$ is located in the same quadrant, then we define
\[
\bz = \bone_{\Gamma} - \bone_{\Gamma^c}\in\RR^n, \quad z_i = \begin{cases}
1, & i \in\Gamma, \\
-1, & i\in\Gamma^c,
\end{cases}
\]
where $\Gamma = \{i: -\pi/4 < \theta_i - \psi < \pi/4\}$ and $\Gamma^c = \{i: 3\pi/4 < \theta_i    -\psi < 5\pi/4\}.$ If so,
\[
\bz^{\top}\BJ \bz = \sum_{1\leq i,j\leq n}a_{ij}\cos(\theta_i - \theta_j)(z_i - z_j)^2 = 8\sum_{i\in\Gamma,j\in\Gamma^c} a_{ij}\cos(\theta_i - \theta_j) < 0
\]
where all $a_{ij}$ are nonnegative with at least one $a_{ij} = 1$ for some $i\in\Gamma$ and $j\in\Gamma^c,$ and $\cos(\theta_i - \theta_j) < 0.$
This implies that the smallest eigenvalue of $\BJ(\btheta)$ is strictly less than 0 while $\BJ(\btheta)$ is positive semidefinite if $\btheta$ is a stable equilibrium, which is a contradiction. 

Now, we wrap up our discussion here: we have shown under~\eqref{eq:main}, if there exists a stable equilibrium, all $\{\theta_i\}_{i=1}^n$ are located in the same quadrant, which means
\[
|\theta_i - \theta_j| < \frac{\pi}{2}, \quad\forall i\neq j,
\]
following from triangle inequality. 
\end{proof}

\subsection{Existence of an equilibrium within the cohesive region}
\begin{lemma}\label{lem:basin}
If $\{\theta_i\}_{i=1}^n$ are initialized in $\Delta(\gamma)$ with $\gamma < \pi/2$, then they will stay inside $\Delta(\gamma)$ if
\begin{equation}\label{eq:basin}
K\sin(\gamma) > \frac{\omega_{\max} - \omega_{\min}}{2\mu-1}, \qquad \mu > \frac{1}{2}.
\end{equation}
In other words, $\Delta(\gamma)$ is a basin of attraction.

\end{lemma}

\begin{proof}
The main idea of proof is adapted from that in~\cite{DB11}. We aim to show that under~\eqref{eq:basin} stated in the lemma, the range of the phase oscillators will not increase if they are initialized in the cohesive region $\Delta(\gamma).$

First of all, we assume there exists an arc of length $\gamma$ covering all $\{\theta_i\}_{i=1}^n$. 
Let $i = \argmax_{1\leq k\leq n} \theta_k$ and $j = \argmin_{1\leq k\leq n}\theta_k$. Note that $i$ and $j$ are unnecessarily unique, i.e., there could be multiple indices such that the maximum and minimum are attained respectively. 
For any $(i, j)$, we have
\[
\gamma = \theta_{i} - \theta_{j} > 0
\]
by definition where $\gamma$ is the range of all oscillators. We will investigate how $\theta_{i} - \theta_{j}$ changes with respect to the time $t$. From~\eqref{model:kura2}, we have
\begin{equation}\label{eq:dottheta_diff}
\dot{\theta}_{i}(t) - \dot{\theta}_{j}(t) = \omega_{i} - \omega_{j} - \frac{K}{n}\sum_{k=1}^n \left(a_{{i}k}\sin(\theta_i - \theta_k) + a_{{j}k}\sin(\theta_k - \theta_{j})\right).
\end{equation}
The goal is to find an upper bound for~\eqref{eq:dottheta_diff}. We claim that
\begin{equation}\label{eq:claim_dotdiff}
\dot{\theta}_{i}(t) - \dot{\theta}_{j}(t)  \leq \omega_{\max} - \omega_{\min} - (2\mu - 1)K\sin(\gamma).
\end{equation}
Note that $\omega_i - \omega_j \leq \omega_{\max} - \omega_{\min}$ holds, and thus the key is to bound the third term in~\eqref{eq:dottheta_diff}.
Apparently, since $\theta_i - \theta_k$ and $\theta_k-\theta_j$ are in $[0,\gamma]\subseteq [0,\pi/2)$, the third term is nonnegative. More precisely, we can improve it by using the fact that each node has at least  $\mu(n-1)$ neighbors. 

{\bf Case 1:} If $a_{ij} = 1$, then
\begin{align*}
& \sum_{k=1}^n \left(a_{{i}k}\sin(\theta_i - \theta_k) + a_{{j}k}\sin(\theta_k - \theta_{j})\right) \\
& \qquad = \sum_{k\neq i,j} \left(a_{{i}k}\sin(\theta_i - \theta_k) + a_{{j}k}\sin(\theta_k - \theta_{j})\right)  + 2\sin(\gamma) \\
& \qquad \geq \sum_{\{k:~a_{{i}k} = a_{{j}k} = 1\}\backslash\{i,j\}} (\sin(\theta_i - \theta_k) + \sin(\theta_k - \theta_j)) + 2\sin(\gamma) 
\end{align*}
where $\sin(\theta_i - \theta_j) = \gamma.$ Note that 
\[
|\{k:~a_{{i}k} = a_{{j}k} = 1\}\backslash\{i,j\}| \geq 2(\mu(n-1) - 1) - (n-2) = (2\mu - 1)n - 2\mu
\]
since $i$ and $j$ share at least $(2\mu - 1)n-2\mu$ neighbors
and
\begin{equation}\label{eq:sine}
\min_{0\leq \theta\leq \gamma} [\sin(\theta) + \sin(\gamma - \theta) ] \geq \sin(\gamma), \quad \gamma \leq \frac{\pi}{2}.
\end{equation}
Thus
\[
\sum_{k=1}^n \left(a_{{i}k}\sin(\theta_i - \theta_k) + a_{{j}k}\sin(\theta_k - \theta_{j})\right)  \geq ((2\mu - 1)n - 2\mu )\sin(\gamma) + 2\sin(\gamma) \geq (2\mu-1)n\sin(\gamma).
\]

{\bf Case 2:}  If $a_{ij} = 0$, then
\begin{align*}
& \sum_{k=1}^n \left(a_{{i}k}\sin(\theta_i - \theta_k) + a_{{j}k}\sin(\theta_k - \theta_{j})\right) \\
& \qquad \geq \sum_{\{k:~a_{{i}k} = a_{{j}k} = 1\}\backslash\{i,j\}} (\sin(\theta_i - \theta_k) + \sin(\theta_k - \theta_j))  \\
& \qquad \geq |  \{k:~a_{{i}k} = a_{{j}k} = 1\}\backslash\{i,j\} | \cdot \sin(\gamma) \\
& \qquad\geq (2\mu(n-1) - (n-2)) \cdot \sin(\gamma) \\
& \qquad \geq (2\mu - 1)n\sin(\gamma).
\end{align*}

Substituting this bound back into~\eqref{eq:dottheta_diff} gives rise to~\eqref{eq:claim_dotdiff}:
\[
\dot{\theta}_{i}(t) - \dot{\theta}_{j}(t) \leq \omega_i  -\omega_j - (2\mu - 1)K\sin(\gamma) \leq \omega_{\max}  -\omega_{\min} - (2\mu - 1)K\sin(\gamma)
\]
which means $\theta_i(t) - \theta_j(t)$ has a negative derivative at $t$ if
\[
\omega_{\max} - \omega_{\min} < (2\mu - 1)K\sin(\gamma)\Longleftrightarrow K\sin(\gamma) > \frac{\omega_{\max} - \omega_{\min}}{2\mu-1}.
\]
This implies that if the condition above holds, the range of $\{\theta_i\}_{i=1}^n$ will not increase, i.e., $\btheta$ stays in $\Delta(\gamma)$ for $\gamma < \pi/2.$
\end{proof}

\begin{lemma}\label{lem:global}
Under the assumption~\eqref{eq:basin}, if $\btheta(0)\in\Delta(\gamma)$ for $\gamma<\frac{\pi}{2}$, then $\btheta(t)$ satisfies exponential frequency synchronization and the frequency synchronized solution lies in $\Delta(\gamma)$.
\end{lemma}

\begin{proof}
First note that Lemma~\ref{lem:basin} implies that as long as the initialization satisfies $\btheta(0)\in\Delta(\gamma)$, $\btheta(t)$ stays in $\Delta(\gamma)$. Now it is safe to only analyze the dynamics in $\Delta(\gamma).$

We consider the $\ell_2$-norm of $\dot{\btheta}(t)$ with initialization $\btheta(0)$ inside the cohesive region. Take the derivative w.r.t. $t$ and we have
\[
\frac{\diff}{\diff t}\|\dot{\btheta}(t)\|^2 =2 \dot{\btheta}^{\top}\ddot{\btheta}
\]
where
\[
\ddot{\theta}_i(t) = -\frac{K}{n}\sum_{j=1}^n a_{ij}\cos(\theta_i(t) - \theta_j(t))\cdot (\dot{\theta}_i(t) - \dot{\theta}_j(t)).
\]
Writing it in matrix form gives
\[
\ddot{\btheta}(t) = -\frac{K}{n} \BJ(\btheta(t)) \dot{\btheta}(t)
\]
where 
$\BJ(\btheta(t))$ satisfies 
\[
\BJ(\btheta) = (J(\btheta))_{ij}, \quad (J(\btheta))_{ij} =  -a_{ij}\cos(\theta_i - \theta_j) \leq - a_{ij}\cos(\gamma) \leq 0
\]
if $\btheta\in\Delta(\gamma)$ and $\gamma < \pi/2.$ 
Note that the underlying network is connected since the degree of each node is at least $n/2$.
Thus the second smallest eigenvalue of $
\BJ(0)$, the Laplacian of the corresponding adjacency matrix, is strictly positive, which is a classical result in spectral graph theory~\cite{Chung97}.
Therefore, $\BJ(\btheta)$ is positive semidefinite for $\btheta\in\Delta(\gamma)$ and its second smallest eigenvalue satisfies
\[
\lambda_2(\BJ(\btheta)) \geq \cos(\gamma)\cdot \lambda_2(\BJ(0)) > 0.
\]
Note that $\dot{\btheta}(t)^{\top}\bone_n = 0$, and thus it holds that
\begin{align*}
\frac{\diff}{\diff t}\|\dot{\btheta}(t)\|^2  & = 2\dot{\btheta}(t)^{\top}\ddot{\btheta}(t)= -\frac{2K}{n}\dot{\btheta}(t)^{\top} \BJ(\btheta(t)) \dot{\btheta}(t) \\
& \leq -\frac{2K}{n} \lambda_2(\BJ(\btheta(t))) \cdot \|\dot{\btheta}(t) \|^2 \\
&  \leq -\frac{2K}{n} \cos(\gamma)\cdot \lambda_2(\BJ(0)) \cdot \|\dot{\btheta}(t)\|^2.
\end{align*}
As a result, we have
\[
\frac{\diff }{\diff t}\left( e^{2n^{-1}K \cos(\gamma) \lambda_2(\BJ(0)) \cdot t } \|\dot{\btheta}(t)\|^2 \right)\leq 0.
\]
By integrating the time $t$ from 0, it holds that
\[
\|\dot{\btheta}(t)\|^2 \leq e^{-2n^{-1}K \cos(\gamma) \lambda_2(\BJ(0))  \cdot t} \|\dot{\btheta}(0)\|^2.
\]
In other words, $\|\dot{\btheta}(t)\|^2$ converges to 0 exponentially fast as $t\rightarrow\infty$ where $\gamma < \frac{\pi}{2}$, i.e., converging to a frequency synchronized solution of the system. Moreover, this equilibrium is stable since it is within the cohesive region and $\BJ(\btheta)\succeq 0$.
\end{proof}

\begin{proof}[Proof of Theorem~\ref{thm:main}]
The proof of the main theorem follows from combining the aforementioned three results, Lemma~\ref{lem:quad}, Lemma~\ref{lem:basin}, and Lemma~\ref{lem:global} together. 
First of all, Lemma~\ref{lem:basin}  implies that if 
\[
\frac{\|\omega\|_{\infty}}{K} \leq \frac{2\mu-1}{2},
\] 
then the cohesive region $\Delta(\pi/2)$ is a basin of attraction. Then Lemma~\ref{lem:global} shows that as long as an initialization is chosen within $\Delta(\pi/2)$, the oscillators will stay in the cohesive region and enjoy exponential frequency synchronization.

On the other hand, suppose
\[
\frac{\|\omega\|_{\infty}}{K}  \leq \sqrt{\mu - \frac{3}{4} + \frac{1-\mu}{n}}  + \mu - 1.
\]
we have shown that all the stable equilibria, if there is any, must stay within $\Delta(\pi/2).$ 

As long as $1\geq \mu > \frac{3-\sqrt{2}}{2}$, we have
\[
 \sqrt{\mu - \frac{3}{4}}  + \mu - 1 \leq \mu - \frac{1}{2}.
\]
This means the condition in Lemma~\ref{lem:quad} is stronger than those in Lemma~\ref{lem:basin} and~\ref{lem:global}. In other words, if~\eqref{eq:main} holds, the whole system has a unique and stable frequency synchronized solution which is also phase-cohesive.

\end{proof}


\end{document}